\pdfoutput=1
\documentclass[times,authoryear]{elsarticle}
\usepackage{amsmath}
\usepackage{amsthm}
\usepackage{amssymb}
\usepackage{amsfonts}

\newtheorem{theorem}{Theorem}
\newtheorem{lemma}[theorem]{Lemma}
\newtheorem{corollary}[theorem]{Corollary}

\newtheorem{proposition}[theorem]{Proposition}
\newtheorem{definition}[theorem]{Definition}

\newtheorem{example}[theorem]{Example}

\usepackage{tikz}
\usetikzlibrary{arrows}
\usepackage{caption}
\usepackage{mathtools}
\usepackage[OT2,T1]{fontenc}

\usepackage{microtype}
\usepackage[section]{algorithm}
\usepackage[noend]{algpseudocode}
\usepackage[pdftex,pagebackref=false,hidelinks]{hyperref}

\let\brace\undefined

\DeclarePairedDelimiter{\abs}{\lvert}{\rvert}

\DeclarePairedDelimiter{\floor}{\lfloor}{\rfloor}

\DeclarePairedDelimiter{\brace}{\lbrace}{\rbrace}
\DeclarePairedDelimiter{\paren}{\lparen}{\rparen}

\newcommand{\C}{\mathbb{C}}
\newcommand{\Z}{\mathbb{Z}}
\newcommand{\T}{\mathsf{T}}
\newcommand{\F}{\mathsf{F}}
\newcommand{\FF}{\mathbb{F}}
\newcommand{\Aeven}{{A_{\mathrm{even}}}}
\newcommand{\Aodd}{{A_{\mathrm{odd}}}}
\newcommand{\Beven}{{B_{\mathrm{even}}}}
\newcommand{\Bodd}{{B_{\mathrm{odd}}}}
\newcommand{\Leven}{L_{\text{even}}}
\newcommand{\Lodd}{L_{\text{odd}}}
\newcommand{\LA}{L_{\text{A}}}
\let\Re\undefined
\let\Im\undefined
\DeclareMathOperator{\resum}{resum}
\DeclareMathOperator{\imsum}{imsum}
\DeclareMathOperator{\Re}{Re}
\DeclareMathOperator{\Im}{Im}

\DeclareMathOperator{\DFT}{DFT}

\renewcommand{\today}{July 27, 2019}
\makeatletter
\def\ps@pprintTitle{%
     \let\@oddhead\@empty
     \let\@evenhead\@empty
     \def\@oddfoot{\footnotesize\itshape
       To appear in the \ifx\@journal\@empty Journal of Symbolic Computation
       \else\@journal\fi\hfill\today}%
     \let\@evenfoot\@oddfoot}
\makeatother

\begin{document}

\begin{frontmatter}

\title{Complex Golay Pairs up to Length $28$: A Search via Computer Algebra and Programmatic SAT}

\author[uw]{Curtis Bright}
\author[wlu]{Ilias Kotsireas}
\author[uw]{Albert Heinle}
\author[uw]{Vijay Ganesh}
\address[uw]{University of Waterloo}
\address[wlu]{Wilfrid Laurier University}

\begin{abstract}
We use techniques from the fields of computer algebra and satisfiability checking
to develop a new algorithm to search for complex Golay pairs.
We implement this algorithm and use it to
perform a complete search for complex Golay pairs of lengths up to~$28$.
In doing so, we find that complex Golay pairs exist in the lengths $24$ and $26$ but
do not exist in the lengths $23$, $25$, $27$, and $28$.
This independently verifies work done by F.~Fiedler in 2013 
and confirms the 2002 conjecture of Craigen, Holzmann, and Kharaghani 
that complex Golay pairs of length $23$ don't exist.
Our algorithm is based on the recently proposed SAT+CAS paradigm 
of combining SAT solvers with computer algebra systems to efficiently search large
spaces specified by both algebraic and logical constraints.
The algorithm has two stages:
first, a fine-tuned computer program uses functionality
from computer algebra systems and numerical libraries to construct a list containing every sequence
that could appear as the first sequence in a complex Golay pair up to equivalence.
Second, a programmatic SAT solver constructs every sequence (if any)
that pair off with the sequences constructed in the first stage
to form a complex Golay pair.
This extends work originally presented at the
International Symposium on Symbolic and Algebraic Computation (ISSAC) in 2018; 
we discuss and implement several improvements to our algorithm that enabled us to improve the
efficiency of the search and increase the maximum length
we search from length $25$ to $28$.
\end{abstract}


\end{frontmatter}

\section{Introduction}

The sequences that are now referred to \emph{Golay sequences} or \emph{Golay pairs} were first 
introduced by \cite{golay1949multi} in a groundbreaking paper on multislit spectrometry.
Later, \cite{golay1961complementary} made a detailed study of their elegant mathematical
properties and in this paper he referred to them as \emph{complementary series}:
\begin{quote}
Regardless of past or possible future applications, the writer has found these complementary
series mathematically appealing\dots
\end{quote}
Since then, Golay pairs and their generalizations have been widely studied
and applied to a huge and surprisingly varied number of problems in engineering.
For example, they have been applied to
encoding acoustic surface waves~\citep{tseng1971signal}, 
spread-spectrum systems~\citep{krone1984quadriphase}, 
optical time domain reflectometry~\citep{nazarathy1989real}, 
CDMA networks~\citep{seberry2002use}, 
medical ultrasounds~\citep{nowicki2003application}, 
train wheel detection~\citep{donato2004train}, 
radar systems~\citep{li2008construction}, 
and wireless networks \citep{lomayev2017golay}. 
Golay pairs consist of two sequences and the property that defines them
(roughly speaking) is the fact that one sequence's ``correlation''
with itself is the inverse of the other sequence's ``correlation'' with itself;
see Definition~\ref{def:cgp} in Section~\ref{sec:background} for
the formal definition.

Although Golay defined his complementary series over an alphabet of $\brace{\pm1}$, later authors have generalized
the alphabet to include nonreal roots of unity such as 
the fourth root of unity $i\coloneqq\sqrt{-1}$.
In this paper, we focus on the case where the alphabet
is $\brace{\pm1,\pm i}$. 
In this case the resulting sequence pairs
are sometimes referred to as \emph{quadriphase}, \emph{4-phase}, or \emph{quaternary} Golay pairs though we
will simply refer to them as \emph{complex} Golay pairs.
If a complex Golay pair of length $n$ exists then we say that $n$ is a \emph{complex Golay number}.

Complex Golay pairs have been extensively studied by many authors
in several different contexts.
Initially they were studied in the context of signal processing
where some researchers ran searches for ``{polyphase complementary codes}''
in the process of studying signal encoding methods.
In particular,
\cite{sivaswamy1978multiphase} found
a complex Golay pair of length~$3$ and \cite{frank1980polyphase}
found complex Golay pairs of length~$5$ and~$13$
and stated that $7$, $9$, $11$, $15$, and $17$ were not complex Golay numbers.
Complex Golay pairs were also introduced by \cite{C:JCMCC:1994}
(independently of the above works)
in order to expand the orders of Hadamard matrices attainable via ordinary Golay pairs---%
that are only known to exist in the lengths $2^{a+b+c}\cdot5^b\cdot13^c$
for integers $a,b,c\geq0$.
Craigen also proved that if $m$ and $n$ are complex Golay numbers
then $2mn$ is a complex Golay number.

An exhaustive search for complex Golay pairs up to length~$13$ was performed by \cite{HK:AJC:1994}.
This search verified the remark of \cite{frank1980polyphase}
that $7$ and $9$ are not complex Golay numbers but
found that $11$ \emph{is} in fact a complex Golay number
(the cause of the mistaken claim is unknown
since it was based on computer programs that were never published).
Later, an exhaustive search up to length~$19$ was performed by \cite*{CHK:DM:2002},
showing that $14$, $15$, $17$, and~$19$ are not complex Golay numbers.
In addition,
they reported
that~$21$ was not a complex Golay number
and conjectured that~$23$ was not a complex Golay number
based on a partial search for complex Golay pairs of length~$23$.

Another line of research is to provide explicit
constructions for complex Golay pairs and a number of results of this
form have been published as well.
\cite*{fiedler2008framework,fiedler2008multi} provide a construction that explains
the existence of all known complex Golay pairs whose lengths are a power
of~$2$,
including complex Golay pairs of length~$16$ discovered by \cite{li2005more} to
not fit into a construction given by \cite{davis1999peak}.
\cite{gibson2011quaternary} provide
a construction that explains the existence of all complex Golay pairs
up to length~$26$ and give a table that lists the total number of complex Golay pairs up to length~$26$.
This table was produced by the mathematician Frank Fiedler,
who describes his enumeration method in a subsequent paper~\citep{fiedler2013small}
where he also reports that~$27$ and~$28$ are not complex Golay numbers.

In this paper we give an enumeration method that can be used to
verify Fiedler's table giving the total number of complex Golay pairs
up to length $28$.
We implemented our method 
and obtained counts up to length~$28$ after
about $8.5$ months of CPU time.
The counts we obtain match those in Fiedler's table in each case, increasing the
confidence that the enumeration was performed without error.
In addition, we also provide counts
for the total number of complex Golay pairs up to well-known equivalence operations
and explicitly publish the sequences
on our website
\href{https://uwaterloo.ca/mathcheck}{\nolinkurl{uwaterloo.ca/mathcheck}}. 
To our knowledge, this is the first time 
that explicit complex Golay pairs (and their counts up to equivalence) have been
published for lengths larger than $19$.
Lastly, we publicly release our code
for enumerating complex Golay pairs
so that others may verify and reproduce our work; we were not able to find
any other code for enumerating complex Golay pairs that was publicly available.

Our result is of interest not only because of the verification we provide
but also because of the method we use to perform the verification.
The method proceeds in two stages.  In the first stage,
a fine-tuned computer program performs an exhaustive search
among all sequences that could possibly appear as the first
sequence in a complex Golay pair of a given length
(up to the equivalence defined in Proposition~\ref{prop:equivGolay} of Section~\ref{sec:background}).
Several filtering theorems that we
describe in Section~\ref{sec:background} allow us to discard almost all
sequences from consideration.  To apply these filtering theorems we use
functionality from the computer algebra system \textsc{Maple}~\citep{Maple10} and the
mathematical library FFTW~\citep{frigo2005design} as we describe in Section~\ref{sec:method}.
After this filtering is completed we have
a list of sequences of a manageable size such that 
the first sequence of every complex Golay pair of a given length
(up to equivalence) appears in the list.

In the second stage, we use the programmatic SAT solver~\textsc{MapleSAT}
developed by \cite{liang2017empirical}
to determine the sequences from the first stage (if any) that can be
paired up with another sequence to form a complex Golay pair.  A
programmatic SAT
instance is constructed from each sequence found in the first stage
such that the instance is satisfiable if and only if the
sequence is part of a complex Golay pair.  Furthermore, when the
instance is satisfiable the assignment produced by the SAT
solver determines the second sequence of a complex Golay pair.

The ``SAT+CAS'' method that 
we use is of interest in its own right because it links the two previously
separated fields of symbolic computation and satisfiability checking.
Recently there has been interest in combining methods from both fields
to solve computational problems as outlined in the
invited talk at ISSAC by~\cite{abraham2015building} and demonstrated by the
SC$^2$ (satisfiability checking + symbolic computation) project initiated by \cite{sc2}.  Our work fits into
this paradigm and to our knowledge is the first application of a SAT solver
to search for complex Golay pairs,
though previous work exists that uses
a SAT solver to search for other types of
complementary sequences
like those defining Williamson matrices
\citep{bright2016mathcheck,DBLP:journals/jar/ZulkoskiBHKCG17,brightthesis}.

A preliminary version of our algorithm and results
was presented at the conference ISSAC~\citep{bright2018enumeration}.
We have since made several improvements to our algorithm
that allow us to extend our exhaustive search
from all lengths $n\leq25$ 
to all lengths $n\leq28$.
We describe our implementation, give timings for our searches,
and compare the timings with those of our previous algorithm
in Section~\ref{sec:results}.
The new timings are about an order of magnitude faster
on the same hardware,
demonstrating the effectiveness of our improvements.
In particular, our new algorithm incorporates the following
updates:

\paragraph{More efficient filtering}
The preprocessing stage of our algorithm
uses a filtering condition (Corollary~\ref{cor:fiedler} of Section~\ref{sec:background})
to remove a large number of sequences from consideration.
In our original algorithm we applied this condition using
a large number of equally-spaced points $z$ along the unit circle.
In our new algorithm we apply this condition significantly
less often but filter approximately the same number
of sequences (and sometimes even more)
by showing how to only use the condition
for the $z$ that are the most likely to work.

\paragraph{More efficient joining}
The first stage of our algorithm requires searching
through two lists and finding elements of the first that
can be joined with elements of the second.  In our
previous algorithm this was done in a totally
brute-force manner, i.e., every item in the first list
was paired with every item of the second list to
see if they could be joined.
In our new algorithm we sort the lists using a custom
ordering and show how to find the
elements that can be joined via a number of linear
scans through the sorted lists (see Section~\ref{sec:matching}),
thereby eliminating
the need to consider all possible pairs.

\paragraph{Increased filtering}
Because the improved filtering method outlined above
is so much faster than the previous filtering method
we were able to add another round of filtering just
before generating the SAT instances (see Section~\ref{sec:optimization}).
The result is that approximately 40\% of the SAT instances
generated by our previous algorithm are now quickly filtered
before even calling a SAT solver. \\

Additionally we include more background details
in this version of the paper.  For example,
we add Section~\ref{sec:altdef} that gives an alternative
definition of complex Golay pairs that is often useful
and we show that this definition is equivalent with the first definition.

\section{Background on Complex Golay Pairs}\label{sec:background}

In this section we present the background necessary to describe our
method for enumerating complex Golay pairs.
First, we require some preliminary definitions
to define what complex Golay pairs are.
We let $\overline{z}$ denote the complex conjugate of $z$
(which is just the multiplicative inverse of $z$
for $z$ on the unit circle) and if $A$ is a sequence we let
$\overline{A}$ denote the sequence containing the conjugates
of $A$.

\begin{definition}[cf.~\cite{kotsireas2013algorithms}]
The \emph{nonperiodic autocorrelation function}
of a sequence $A = [a_0, \dotsc, a_{n-1}]$ of length $n$ is 
\[ N_A(s) \coloneqq \sum_{k=0}^{n-s-1} a_k\overline{a_{k+s}} \qquad\text{for}\qquad s=0,\dotsc, n-1 . \]
\end{definition}

\begin{definition}\label{def:cgp}
  A pair of sequences $(A,B)$ with $A$ and $B$ in $\{\pm1,\pm i\}^n$
  are called a \emph{complex Golay pair}
  if the sum of their nonperiodic autocorrelations is a constant zero for $s\neq0$, i.e.,
  \begin{equation*}
  N_A(s) + N_B(s) = 0 \qquad\text{for}\qquad s = 1, \dotsc, n-1 . 
  \end{equation*}
\end{definition}

Note that if $A$ and $B$ are in $\{\pm1,\pm i\}^n$ (as we assume throughout this paper)
then $N_A(0)+N_B(0)=2n$ by the definition of
the nonperiodic autocorrelation function
and the fact that
$z\overline{z}=1$ if $z$ is $\pm1$ or $\pm i$, explaining why $s\neq0$ in Definition~\ref{def:cgp}.
\begin{example}
$([1,1,-1], [1, i, 1])$ is a complex Golay pair of length $3$ since the
first sequence $A$ has autocorrelations $N_A(1)=0$ and $N_A(2)=-1$
and the second sequence $B$ has autocorrelations $N_B(1)=0$ and $N_B(2)=1$.
\end{example}

\subsection{Alternative definition}\label{sec:altdef}

Instead of viewing complex Golay pairs as pairs of \emph{sequences}
it is also possible to view them as pairs of \emph{polynomials}.
If $A$ is the sequence $[a_0,\dotsc,a_{n-1}]$ we let $A(z)$ denote the
\emph{Hall polynomial} $a_0+a_1z+\dotsc+a_{n-1}z^{n-1}$ which can be
viewed as the finite generating function of $A$.
This leads to the following alternative definition of complex Golay pairs.

\begin{definition}\label{def:altcgp}
A pair of polynomials $(A(z),B(z))$ that have degrees $n-1$ and coefficients in $\{\pm1,\pm i\}$
are called a \emph{complex Golay pair} if\/ $\abs{A(z)}^2 + \abs{B(z)}^2 = 2n$
for all $z$ on the unit circle.
\end{definition}

\begin{example}
$(1+z-z^2,1+iz+z^2)$ is a complex Golay pair of length $3$ since
\[ \abs[\big]{1+z-z^2}^2 + \abs[\big]{1+iz+z^2}^2 = 6 \]
for all $z$ on the unit circle.
\end{example}

These two definitions can be seen to be equivalent using the following lemma
that expresses the squared absolute values of a polynomial
evaluated on the unit circle in terms of the autocorrelations of the sequence
formed by the polynomial's coefficients.

\begin{lemma}\label{lem:normeval}
Let $A$ be a complex sequence of length $n$.  Then
\[ \abs{A(z)}^2 = N_A(0) + 2\Re\paren[\bigg]{\sum_{s=1}^{n-1}{N_A(s)z^{-s}}}. \]
for all $z$ on the unit circle.
\end{lemma}

\begin{proof}
Since $z$ is on the unit circle we have $\overline{z}=z^{-1}$, so $\abs{A(z)}^2=A(z)\overline{A(z)}=A(z)\overline{A}(z^{-1})$.
Expanding this, we obtain that $\abs{A(z)}^2$ is equal to
\[ \sum_{k=0}^{n-1}\sum_{l=0}^{n-1} a_k \overline{a_l} z^{k-l} . \]
Collecting terms of $z^s$ together for $s$ from $-n+1$ to $n-1$ this becomes
\[ \sum_{s=0}^{n-1}\paren[\bigg]{\sum_{j=s}^{n-1}a_j\overline{a_{j-s}}}z^s + \sum_{s=1}^{n-1}\paren[\bigg]{\sum_{j=s}^{n-1}a_{j-s}\overline{a_{j}}}z^{-s} \]
so the coefficient of $z^s$ for positive $s$ is $\overline{N_A(s)}$ and for negative $s$ is $N_A(s)$.
Using the fact that $\overline{N_A(s)}z^s=\overline{N_A(s)z^{-s}}$ 
this becomes
\[ N_A(0) + \sum_{s=1}^{n-1}\paren[\bigg]{\overline{N_A(s)z^{-s}} + N_A(s)z^{-s}} . \]
The desired result now follows from the fact that $z+\overline{z}=2\Re(z)$ for all complex $z$.
\end{proof}

For completeness we now demonstrate the equivalence of Definitions~\ref{def:cgp} and~\ref{def:altcgp}.

\begin{theorem}
A pair $(A,B)$ is a complex Golay pair in the
sense of Definition~\ref{def:cgp} if and only if
$(A(z),B(z))$ is a complex Golay pair in the sense
of Definition~\ref{def:altcgp}.
\end{theorem}

\begin{proof}
If $(A,B)$ is a complex Golay pair in the
sense of Definition~\ref{def:cgp} then
by Lemma~\ref{lem:normeval}
\[ \abs{A(z)}^2 + \abs{B(z)}^2 = N_A(0) + N_B(0) + 2\Re\paren[\bigg]{\sum_{s=1}^{n-1}{(N_A(s)+N_B(z))z^{-s}}} = 2n \]
for all $z$ on the unit circle, as required.

Conversely, if $(A(z),B(z))$ is a complex Golay pair in the sense of Definition~\ref{def:altcgp} then
by Lemma~\ref{lem:normeval} and subtracting off $N_A(0) + N_B(0) = 2n$ we have
\[ \sum_{s=1}^{n-1}{\paren[\Big]{N_A(s)z^{-s}+\overline{N_A(s)}z^{s}+N_B(z)z^{-s}+\overline{N_B(s)}z^{s}}} = 0 \]
for all $z$ on the unit circle.  Since any nonzero rational function has a finite number of zeros in $\C$
the function given on the left hand side must be identically zero.  
In particular, the coefficients of $z^{-s}$ for $s=1$, $\dotsc$, $n-1$ must be zero
and we derive $N_A(s)+N_B(s)=0$ for $s=1$, $\dotsc$, $n-1$ as required.
\end{proof}

\subsection{Equivalence operations}\label{subsec:equiv}

There are certain invertible operations that preserve the property of
being a complex Golay pair when applied to a sequence pair $(A,B)$.
These are summarized in the following proposition.

\begin{proposition}[cf.~\cite{CHK:DM:2002}]
\label{prop:equivGolay}
Let\/ $([a_0,\ldots,a_{n-1}]$, $[b_0,\ldots,b_{n-1}])$ be a complex
Golay pair. Then the following are also complex Golay pairs:
\begin{enumerate}
\item[E1.] (Reversal)\/ $([a_{n-1},\dotsc,a_0], [b_{n-1},\dotsc,b_0])$.
\item[E2.] (Conjugate Reverse $A$)\/ $([\overline{a_{n-1}},\ldots,\overline{a_0}], [b_0,\ldots,b_{n-1}])$.
\item[E3.] (Swap)\/ $([b_0,\dotsc,b_{n-1}], [a_0,\dotsc,a_{n-1}])$. 
\item[E4.] (Scale $A$)\/ $([ia_0,\dotsc,ia_{n-1}], [b_0,\dotsc,b_{n-1}])$. 
\item[E5.] (Positional Scaling)\/ $(i\star A, i\star B)$
where $c\star A$ denotes the sequence of coefficients of the polynomial $A(cz)$, i.e.,
$[a_0,ca_1,c^2a_2,\dotsc,c^{n-1}a_{n-1}]$.
\end{enumerate}
\end{proposition}

\begin{proof}Suppose $(A,B)$ is a complex Golay pair and let $z$ be on the unit circle.
\begin{enumerate}
\item[E1.] Note that the reverse of $A(z)$ is $z^{n-1}A(z^{-1})$ and $\abs{z^{n-1}}=1$.  Then
\[ \abs{z^{n-1}A(z^{-1})}^2 + \abs{z^{n-1}B(z^{-1})}^2 = \abs{A(z^{-1})}^2 + \abs{B(z^{-1})}^2 = 2n \]
since $(A,B)$ is a complex Golay pair and $z^{-1}$ is on the unit circle.
\item[E2.] The conjugate reverse of $A(z)$ is $z^{n-1}\overline{A}(z^{-1})$ and $\abs{z^{n-1}\overline{A}(z^{-1})}=\abs{\overline{A(z)}}=\abs{A(z)}$.
\item[E3.] $\abs{B(z)}^2+\abs{A(z)}^2=\abs{A(z)}^2+\abs{B(z)}^2=2n$.
\item[E4.] $A(z)$ scaled by $i$ is $iA(z)$ and $\abs{iA(z)}=\abs{A(z)}$.
\item[E5.] $\abs{A(iz)}^2+\abs{B(iz)}^2=2n$ since $(A,B)$ is a complex Golay pair and $iz$ is on the unit circle.
\end{enumerate}
\end{proof}

\begin{definition}
We call two complex Golay pairs $(A, B)$ and $(A', B')$ \emph{equivalent}
if $(A', B')$ can be obtained from
$(A, B)$ using the transformations described in Proposition~\ref{prop:equivGolay}.
\end{definition}

The next lemma provides some normalization conditions that can be used when searching
for complex Golay pairs up to equivalence.
Since all complex Golay pairs $(A',B')$ which are equivalent to
a complex Golay pair $(A,B)$ can easily
be generated from $(A,B)$, it suffices to search for complex Golay pairs up to equivalence.

\begin{lemma}[cf.~\cite{fiedler2013small}]\label{lem:normalize}
  Let\/ $(A', B')$ be a complex Golay pair.
  Then\/ $(A', B')$ is equivalent to a complex Golay pair\/
  $(A, B)$ with\/ $a_0=a_1=b_0=1$ and\/ $a_2\in\brace{\pm1, i}$.
\end{lemma}
\begin{proof}
We will transform a given complex Golay sequence pair $(A', B')$ into 
an equivalent normalized one
using the equivalence operations of
Proposition~\ref{prop:equivGolay}.
To start with, let $A\coloneqq A'$ and $B\coloneqq B'$.

First, we ensure that $a_0=1$.  To do this, we apply operation E4 (scale $A$) enough times
until $a_0=1$.

Second, we ensure that $a_1=1$.  To do this, we apply operation E5 (positional scaling) enough times
until $a_1=1$; note that E5 does not change $a_0$.

Third, we ensure that $a_2\neq -i$.  If it is, we apply operation E1 (reversal) and E2 (conjugate reverse $A$)
which has the effect of keeping $a_0=a_1=1$ and setting $a_2=i$.

Last, we ensure that $b_0=1$.  To do this, we apply operation E3 (swap) and then operation E4 (scale $A$) enough times
so that $a_0=1$ and then operation E3 (swap) again.  This has the effect of not changing $A$ but setting $b_0=1$.
\end{proof}

\subsection{Filtering properties}

We now prove some useful properties that all complex
Golay pairs satisfy and that will be exploited by our algorithm
for enumerating complex Golay pairs.
The properties will be used as filtering criteria:
if a sequence does not satisfy them then we know it cannot
possibly be part of a complex Golay pair and may therefore
be filtered away.
The next well-known lemma is 
one of the most powerful results of this form.

\begin{lemma}[cf.~\cite{popovic1991synthesis}]\label{cor:filter}
Let $A$ be a complex sequence of length $n$.
If\/ $\abs{A(z)}^2 > 2n$ for some $z$ on the unit circle
then $A$ is not a member of a complex Golay pair.
\end{lemma}
\begin{proof}
Suppose the sequence $A$ was a member of a complex Golay pair whose other
member was the sequence $B$.
Since $\abs{B(z)}^2\geq0$, we must have $\abs{A(z)}^2+\abs{B(z)}^2>2n$,
in contradiction to Definition~\ref{def:altcgp}.
\end{proof}

\begin{example}
The sequence $A\coloneqq[1,1,1]$ cannot be a member of a complex Golay pair
since $\abs{A(1)}^2=9$ is larger than $2n=6$.
\end{example}

\cite{fiedler2013small} derives a variation of Lemma~\ref{cor:filter}
that is useful because it can be applied knowing only around half of the entries of $A$.
It is derived using the following variation of Lemma~\ref{lem:normeval}.
Let $\Aeven$ be identical to $A$ with the entries of odd index replaced by zeros and let $\Aodd$ be
identical to $A$ with the entries of even index replaced by zeros.

\begin{lemma}[cf.~\cite{fiedler2013small}]\label{lem:fiedler}
Let $A$ be a complex sequence of length $n$.  Then
\[ \abs{\Aeven(z)}^2 + \abs{\Aodd(z)}^2 = N_A(0) + 2\Re\paren[\Bigg]{\sum_{\substack{s=1\\\text{\rm$s$ even}}}^{n-1}{N_A(s)z^{-s}}}. \]
\end{lemma}

\begin{proof}
The proof proceeds like in the proof of Lemma~\ref{lem:normeval} and we derive
\[ \abs{\Aeven(z)}^2 + \abs{\Aodd(z)}^2 = \sum_{\substack{k,l=0\\\text{$k,l$ even}}}^{n-1} a_k \overline{a_l} z^{k-l} + \sum_{\substack{k,l=0\\\text{$k,l$ odd}}}^{n-1} a_k \overline{a_l} z^{k-l} = \sum_{\substack{k,l=0\\k-l\text{ even}}}^{n-1} a_k \overline{a_l} z^{k-l} . \]
From this we see that the coefficient on $z^s$ for odd $s$ is zero and the coefficient on $z^s$ for even $s$
is the same as in Lemma~\ref{lem:normeval} from which the result follows.
\end{proof}

\begin{corollary}\label{cor:fiedler}
Let $A$ be a complex sequence of length $n$.
If\/ $\abs{\Aeven(z)}^2$ or $\abs{\Aodd(z)}^2$ is strictly larger than $2n$
for some $z$ on the unit circle
then $A$ is not a member of a complex Golay pair.
\end{corollary}

\begin{proof}
If $(A,B)$ is a complex Golay pair then by Lemma~\ref{lem:fiedler}
we have
\[ \abs{\Aeven(z)}^2 + \abs{\Aodd(z)}^2 + \abs{\Beven(z)}^2 + \abs{\Bodd(z)}^2 = 2n \]
but if either $\abs{\Aeven(z)}^2>2n$
or $\abs{\Aodd(z)}^2>2n$ then this cannot hold.
\end{proof}

\begin{example}
The sequence $A\coloneqq[1,x,1,y,1,z,1]$ cannot be a member of a complex Golay pair
(regardless of the values of $x$, $y$, and $z$)
since $\abs{\Aeven(1)}^2=16$ is larger than $2n=14$.
\end{example}

\subsection{Sum-of-squares decomposition types}\label{subseq:decomp}

The next lemma is useful because it allows us to write $2n$
as the sum of four integer squares.  It is stated
by \cite{HK:AJC:1994} using a different notation; we use the notation
$\resum(A)$ and $\imsum(A)$ to represent the real and imaginary parts of
the sum of the entries of $A$.
For example, if $A\coloneqq[1,i,-i,i]$ then $\resum(A)=\imsum(A)=1$. 

\begin{lemma}[cf.~\cite{HK:AJC:1994}]\label{lem:decomp}
If\/ $(A, B)$ is a complex Golay pair then
\[ {\resum(A)}^2 + {\imsum(A)}^2 + {\resum(B)}^2 + {\imsum(B)}^2 = 2n . \]
\end{lemma}
\begin{proof}
Using Definition~\ref{def:altcgp} with $z=1$ we have
\[ \abs{\resum(A)+\imsum(A)i}^2 + \abs{\resum(B)+\imsum(B)i}^2 = 2n . \]
Since $\abs{\resum(X)+\imsum(X)i}^2 = \resum(X)^2 + \imsum(X)^2$ the result follows.
\end{proof}

A consequence of Lemma~\ref{lem:decomp} is
that every complex Golay pair generates a decomposition of $2n$ into a sum of four integer squares.
In fact, it typically generates several decompositions of $2n$ into a sum of four squares.
Recall that $i\star A$ denotes positional scaling by~$i$ (operation E5) on the sequence~$A$.
If $(A,B)$ is a complex Golay pair then applying operation E5 to this pair $k$ times 
shows that $(i^k\star A,i^k\star B)$ 
is also a complex Golay pair.  
By using Lemma~\ref{lem:decomp} on these complex Golay pairs
one obtains the fact that $2n$ can be decomposed as the sum of 
four integer squares as
\begin{equation*}
\resum\paren{i^k\star A}^2 + \imsum\paren{i^k\star A}^2 + \resum\paren{i^k\star B}^2 + \imsum\paren{i^k\star B}^2 .
\end{equation*}
For $k>3$ this produces no new decompositions but in general
for $k=0$, $1$, $2$, and~$3$ this produces four distinct
decompositions of $2n$ into a sum of four squares.

With the help of a computer algebra system (CAS) one can
enumerate every possible way that $2n$ may be written as a sum of four integer squares.  For example,
when $n=23$ one has $0^2+1^2+3^2+6^2=2\cdot23$ and $1^2+2^2+4^2+5^2=2\cdot23$ as well as all
permutations of the squares and negations of the integers being squared.
During the first stage of our enumeration method only the first sequence of a complex Golay pair
is known, so at that stage we cannot compute its whole
sums-of-squares decomposition.
However, it is still possible to filter some sequences from consideration based on analyzing
the two known terms in the sums-of-squares decomposition.

For example, say that $A$ is the
first sequence in a potential complex Golay pair of length $23$ with $\resum(A)=0$ and $\imsum(A)=5$.
We can immediately discard $A$ from consideration because 
there is no way to chose the $\resum$ and $\imsum$ of $B$ to complete the sums-of-squares
decomposition of $2n$, i.e., there are no integer solutions $(x,y)$ of $0^2+5^2+x^2+y^2=2n$.

\section{Enumeration Method}\label{sec:method}

In this section we describe in detail the method we used to
perform a complete enumeration of all complex Golay pairs up to length~$28$.
Given a length $n$ our goal is to find all $\{\pm1,\pm i\}$ sequences $A$ and $B$
of length $n$ such that $(A,B)$ is a complex Golay pair.

\subsection{Preprocessing: Enumerate possibilities for $\Aeven$ and $\Aodd$}\label{sec:preproc}

The first step of our method uses Fiedler's trick of considering
the entries of $A$ of even index separately from the entries of $A$ of odd index.
There are approximately $n/2$ nonzero entries in each of $\Aeven$ and $\Aodd$
and there are four possible values for each nonzero entry.  Therefore there
are approximately $2\cdot4^{n/2}=2^{n+1}$ possible sequences to check in this
step.  Additionally, by Lemma~\ref{lem:normalize} we may assume the first nonzero
entry of both $\Aeven$ and $\Aodd$ is $1$ and that the second nonzero
entry of $\Aeven$ is not $-i$, decreasing the number of sequences
to check in this step by more than a factor of $4$.
It is quite feasible to perform a brute-force search
through all such sequences when $n\approx30$.

We apply Corollary~\ref{cor:fiedler} to every possibility for $\Aeven$ and $\Aodd$.
There are an infinite number of~$z$ on the unit circle so
it is not possible to apply Corollary~\ref{cor:fiedler} using all such~$z$.
One simple approach 
is to try a sufficiently large number of points~$z$ so that 
when a point exists with $\abs{A'(z)}^2>2n$
(where $A'$ is either $\Aeven$ or $\Aodd$) such a point is usually discovered.
For example, \cite{bright2018enumeration} tested Corollary~\ref{cor:fiedler}
for $2^{14}$ equally-spaced points~$z$ around the unit circle.

Instead, in our implementation we use a nonlinear programming method
to estimate the maximum of $\abs{A'(z)}^2$ for $z$ on the unit circle.
This can be done with the \textsc{NLPSolve} command of the
computer algebra system \textsc{Maple}
though for efficiency we use a custom C implementation
of a variant of the quadratic interpolation
method described by \cite{sun2006optimization}.

We write $\abs{A'(z)}^2$ in terms of the real variable
$\theta$ by using the substitution $z=e^{i\theta}$ and define $f(\theta)\coloneqq\abs{A'(e^{i\theta})}^2$.
The quadratic interpolation method to estimate the maximum of $f(\theta)$ over $0\leq\theta<2\pi$ proceeds as follows:
\begin{enumerate}
\item Evaluate $f$ at the $2^7$ points $\theta_0$, $\dotsc$, $\theta_{127}$ where $\theta_k\coloneqq\frac{2\pi k}{128}$
and let $f_k\coloneqq f(\theta_k)$.
If $f_k>2n$ for some $k$ we can immediately filter $A'$ by Corollary~\ref{cor:fiedler}.
\item 
For every value of $f$ that is larger than its neighbours (i.e.,
values of $f_k$ that satisfy $f_{k-1}\leq f_k$ and $f_k\geq f_{k+1}$)
we use interpolation to find the quadratic polynomial that passes through the points
$(\theta_{k-1},f_{k-1})$, $(\theta_{k},f_{k})$, and $(\theta_{k+1},f_{k+1})$.
Let $\theta^*$ be the value of $\theta$ that maximizes the quadratic polynomial;
this can be computed to be exactly
\[ \frac{1}{2}\cdot\frac{f_{k-1}(\theta_k^2-\theta_{k+1}^2)+f_k(\theta_{k+1}^2-\theta_{k-1}^2)+f_{k+1}(\theta_{k-1}^2-\theta_k^2)}{f_{k-1}(\theta_k-\theta_{k+1})+f_k(\theta_{k+1}-\theta_{k-1})+f_{k+1}(\theta_{k-1}-\theta_k)} \]
and let $f^*\coloneqq f(\theta^*)$.
\item Note that $f^*$ is often a better approximation to a local maximum of $f$ than
$f_k$ was, and if $f^* > 2n$ then we can filter $A'$ by Corollary~\ref{cor:fiedler}.
Otherwise we can use the point $(\theta^*,f^*)$ to derive a tighter interval in
which a local maximum of $f$ must lie.
For example, if $\theta_k<\theta^*<\theta_{k+1}$ and $f^*>f_k$ then we
can repeat the previous step except using the points $(\theta_k,f_k)$, $(\theta^*,f^*)$,
and $(\theta_{k+1},f_{k+1})$.
\end{enumerate}

One can use this method to derive more and more accurate approximations
to the local maxima of $f$ though this is mostly unnecessary for our purposes
as we only care about finding a single value for $\theta$ with $f(\theta)>2n$.
A good approximation to a local maximum was usually
found after a single interpolation step so in our implementation
we would move on to looking for another local maximum of $f$
after repeating the interpolation step three times.
If all values of $f_k$ were examined and no
points $\theta$ were found with $f(\theta)>2n$ then
we save $A'$ as a sequence that could not be filtered.

At the conclusion of this step we have two lists: one list $\Leven$
of the $\Aeven$ that were not discarded and one list $\Lodd$ of the $\Aodd$
that were not discarded.

\subsection{Stage 1: Enumerate possibilities for $A$}\label{sec:matching}

We now enumerate all possibilities for $A$ by joining all
possibilities for $\Aeven$ with all possibilities for $\Aodd$.
The most straightforward way of doing this would be to simply try
all $A_1\in\Lodd$ and $A_2\in\Leven$; this was done by \cite{bright2018enumeration}.
However, because both $\Leven$ and $\Lodd$ can contain millions
of sequences it can be inefficient to try all possible pairings
$(A_1,A_2)$.
However, many pairings can be eliminated
by using the conditions implied by Lemma~\ref{lem:decomp};
we now show how to find all possible
pairings that satisfy these conditions without
needing to try all $A_1\in\Lodd$ and $A_2\in\Leven$.

Let $(u_0,u_1,u_2,u_3)$ be an arbitrary quadruple of integers.
We will describe how to efficiently find all $A$
whose entries of odd index are in $\Lodd$, whose entries of even index
are in $\Leven$, and that satisfy
\[ \resum(A)=u_0, \quad \imsum(A)=u_1, \quad \resum(i\star A)=u_2, \quad \imsum(i\star A)=u_3 . \tag{$1$}\label{eq:U} \]
For each $A_1\in\Lodd$ we form the vector 
\[ V_1 \coloneqq (\resum(A_1),\imsum(A_1),\resum(i\star A_1),\imsum(i\star A_1)) \]
and for each $A_2\in\Leven$ we form the vector 
\[ V_2 \coloneqq (u_0-\resum(A_2),u_1-\imsum(A_2),u_2-\resum(i\star A_2),u_3-\imsum(i\star A_2)) . \]
We now show that the $A$ that satisfy relationship~\eqref{eq:U} are exactly those formed
by the $A_1$ and $A_2$ with $V_1=V_2$.

\begin{lemma}\label{lem:matching}
Let\/ $A_1\in\Lodd$ and\/ $A_2\in\Leven$ and let\/ $A$ be the sequence
that satisfies the relationship\/ $A(z)=A_1(z)+A_2(z)$, i.e., $A$ is a sequence
of length\/ $n$ with\/ $\{\pm1,\pm i\}$ entries.
Then\/ $A$ satisfies~\eqref{eq:U} if and only if the vectors\/ $V_1$ and\/ $V_2$
defined as above satisfy\/ $V_1=V_2$.
\end{lemma}

\begin{proof}
Consider the first entry of $V_1$ and $V_2$.  If they are equal, we have
\[ \resum(A_1) = u_0 - \resum(A_2) . \]
This implies that $\resum(A)=u_0$ since $\resum(A)=\resum(A_1)+\resum(A_2)$.
Similarly, we derive $\imsum(A)=u_1$, $\resum(i\star A)=u_2$, and $\imsum(i\star A)=u_3$,
as required.  Conversely, $\resum(A)=u_0$ implies that the first entries of $V_1$
and $V_2$ are equal and the other equations from~\eqref{eq:U} imply that
their other entries are also equal.
\end{proof}

Therefore we have translated the problem of finding the sequences $A$ that satisfy~\eqref{eq:U}
into the problem of finding $A_1\in\Lodd$ and $A_2\in\Leven$
that have matching vectors $V_1$ and~$V_2$.
We can efficiently solve this problem using a string sorting
algorithm as described by (for example) \cite*{kotsireas2009weighing}.
First, we sort the list $\Lodd$ so that the
vectors $V_1$ formed above appear in lexicographically increasing
order when iterating through $L_1\in\Lodd$.  Similarly, we sort
$\Leven$ in the same way so that the $V_2$ appear in lexicographically
increasing order.

We can then find all $V_1$ that match with $V_2$ by a linear
scan through the lists $L_1$ and~$L_2$
and this requires only one pass
through each list.
For example, if at some point we find that $V_1$ is lexicographically greater
than $V_2$ then we try the next $L_2$ in the list $\Leven$ and if
we instead find that $V_2$ is lexicographically greater than $V_1$
then we try the next $L_1$ in the list $\Lodd$.
If instead we find $V_1=V_2$ we iterate through the $L_1$ and $L_2$
that share the same value of $V_1$ and $V_2$ and save the $A$
formed by joining all such $A_1$ and $A_2$ in a new list~$\LA$.

It remains to determine the appropriate values of $(u_0,u_1,u_2,u_3)$
to use in~\eqref{eq:U}.
To do this we use a quadratic Diophantine equation solver and find all solutions
of the Diophantine equation
\[ u^2 + v^2 + x^2 + y^2 = 2n \qquad \text{in integers $u$, $v$, $x$, $y$} . \]
Let $U$ be the set of all pairs $(u,v)$ for which
this equation is solvable.
By Lemma~\ref{lem:decomp} and operation~E5
we know that if $A$ is a complex Golay pair
then $(\resum(A),\imsum(A))$ and $(\resum(i\star A),\imsum(i\star A))$
must be members of $U$.
Therefore, we apply the above joining procedure
for all $((u_0,u_1),(u_2,u_3))\in U^2$.

At the conclusion of this stage we will have a list of sequences $\LA$
that could potentially be a member of a complex Golay pair.
By construction, the first member of all complex Golay pairs
(up to the equivalence described in Lemma~\ref{lem:normalize})
of length $n$ will be in $\LA$.
To decrease the size of $\LA$ even further we perform
additional filtering before adding sequences into $\LA$,
for example, we ensure that
\[ \paren[\big]{\resum(-1\star A),\imsum(-1\star A)} \quad\text{and}\quad \paren[\big]{\resum(-i\star A),\imsum(-i\star A)} \]
are both in $U$ for each $A$ added to $\LA$.  We also
filter those $A$ with $\abs{A(z)}^2 > 2n$
for some~$z$ on the unit circle (see Section~\ref{sec:optimization}
for optimization details).

\subsection{Stage 2: Construct the second sequence $B$ from $A$}\label{sec:stage2}

In the second stage we take as input the list $\LA$ generated
in the first stage, i.e., a list of the sequences $A$ that were not filtered by
any of the filtering theorems we applied.  For each $A\in\LA$ we attempt to
construct a second sequence $B$ such that $(A,B)$ is a complex Golay pair.
We do this by generating a SAT instance that encodes the property of
$(A,B)$ being a complex Golay pair where
the entries of $A$ are known and the entries
of $B$ are unknown and encoded using Boolean variables.
Because there are four possible values for each entry of $B$
we use two Boolean variables to encode each entry. 
Although the exact encoding used is arbitrary, we fixed the following
encoding in our implementation, where the variables $v_{2k}$ and $v_{2k+1}$
represent $b_k$, the $k$th entry of $B$:
\[ \begin{array}{@{\;\;}c@{\;\;}c@{\;\;}|@{\;\;}c@{\;\;}}
v_{2k} & v_{2k+1} & b_k \\ \hline
\F & \F & 1 \\
\F & \T & -1 \\
\T & \F & i \\
\T & \T & -i
\end{array} \]

To encode the property that $(A,B)$ is a complex Golay pair in our SAT instance
we add the conditions that define $(A,B)$ to be a complex Golay pair, i.e.,
\begin{equation*}
N_A(s) + N_B(s) = 0 \qquad\text{for}\qquad s = 1, \dotsc, n-1 . 
\end{equation*}
These equations could be encoded using clauses in conjunctive normal form
(for example by constructing
logical circuits to perform complex multiplication and addition and then
converting those circuits into CNF clauses).  However, we found that a much
more efficient and convenient method was to use a \emph{programmatic} SAT solver.

The concept of a programmatic SAT solver was first introduced
by \cite{ganesh2012lynx} where a programmatic SAT solver was
shown to be more efficient than a standard SAT solver when solving
instances derived from RNA folding problems.
More recently, a programmatic SAT solver was also shown to be
useful when searching for Williamson matrices and good matrices by \cite{bright2018sat,aaai2019preprint}.
Generally, programmatic SAT solvers perform well when there is additional
domain-specific knowledge known about the problem being solved.
Often this knowledge cannot easily be encoded into a SAT instance directly
but can be given to a programmatic SAT solver
to help guide the solver in its search.

Concretely, a programmatic SAT solver is compiled with a piece
of code that encodes a property any solution
must satisfy.  Periodically the SAT solver will run this code
while performing its search, and if the current partial
assignment violates a property that is expressed in the provided
code then a conflict clause is generated encoding this fact.
The conflict clause is added to the SAT solver's database
of learned clauses where it is used to increase the efficiency
of the remainder of the search.
The reason that these clauses can be so useful is because they can
encode facts that the SAT solver would have no way of learning
otherwise, since the SAT solver has no knowledge of the domain of
the problem.

Not only does this paradigm allow the SAT solver to perform its
search more efficiently, it also allows instances
to be much more expressive.  Under this framework SAT instances do not
have to consist solely of Boolean formulas in conjunctive normal
form (the typical format of SAT instances) but can consist of
clauses in conjunctive normal form combined with a piece of
code that \emph{programmatically} expresses clauses.
Increased expressiveness is also a feature of SMT (SAT modulo theories) solvers,
though SMT solvers typically require additional overhead
and only support a fixed number of theories such as those
specified in the SMT library~\citep{BarFT-SMTLIB}.  Additionally,
one can compile \emph{instance-specific} programmatic SAT solvers
that are tailored to perform searches for a specific class of
problems.

For our purposes we use a programmatic
SAT solver tailored to search for
sequences~$B$ that when paired with a given sequence~$A$
form a complex Golay pair.
Each instance will contain the $2n$ variables $v_0$, $\dotsc$, $v_{2n-1}$
that encode the entries of $B$ as previously specified.
In detail, the code given to the SAT solver does the following:
\begin{enumerate}
\item Compute and store the values $N_A(k)$ for $k=1$, $\dotsc$, $n-1$.
\item Initialize $s$ to $n-1$.  This will be a variable that
controls which autocorrelation condition we are currently examining.
\item Examine the current partial assignment to $v_0$, $v_1$,
$v_{2n-2}$, and $v_{2n-1}$.  If all these values have been assigned
then we can determine the values of $b_0$ and $b_{n-1}$.
From these values we compute
$N_B(s) = b_0\overline{b_{n-1}}$.
If $N_A(s)+N_B(s)\neq0$ then $(A,B)$ cannot be a complex Golay pair
(regardless of the values of $b_1$, $\dotsc$, $b_{n-2}$) and therefore
we learn a conflict clause saying that $b_0$ and $b_{n-1}$ cannot
both be assigned to their current values.  More explicitly,
if $v_k^{\text{cur}}$ represents the literal $v_k$ when $v_k$
is currently assigned to true and the literal $\lnot v_k$ when
$v_k$ is currently assigned to false
we learn the clause
\[ \lnot v_0^{\text{cur}}\lor\lnot v_1^{\text{cur}}\lor\lnot v_{2n-2}^{\text{cur}}\lor\lnot v_{2n-1}^{\text{cur}} \]
that says that at least one of $\{v_0, v_1, v_{2n-2}, v_{2n-1}\}$ must have their value changed.
\item Decrement $s$ by $1$ and repeat the previous step,
computing $N_B(s)$ if the all the $b_k$ that appear
in its definition have known values.
If $N_A(s)+N_B(s)\neq0$ then learn a clause preventing the values
of $b_k$ that appear in the definition of $N_B(s)$ from being assigned
the way that they currently are.  Continue to repeat this step until $s=0$.
\item If all values of $B$ are assigned but no clauses
have been learned in steps~3--4 then output the complex Golay pair $(A,B)$.
If an exhaustive search is desired, learn a clause preventing
the values of $B$ from being assigned the way they currently are;
otherwise learn nothing and return control to the SAT solver.
\end{enumerate}
For each $A$ in the list $\LA$ from stage~1 we run a SAT solver with the above
programmatic code. The list of all outputs $(A,B)$ in step~5 shown
above now form a complete list of complex Golay pairs of length $n$
up to the equivalence given in Lemma~\ref{lem:normalize}.  In fact,
since Lemma~\ref{lem:normalize} says that we can set $b_0=1$ we can
assume that both $v_0$ and $v_1$ are always set to false.
In other words, we can add the two unit clauses $\lnot v_0$
and $\lnot v_1$ into our SAT instance without omitting any
complex Golay pairs up to equivalence.

\subsection{Postprocessing: Enumerating all complex Golay pairs}

At the conclusion of the second stage we have obtained a list of complex
Golay pairs of length $n$ such that every complex Golay pair of length $n$
is equivalent to some pair in our list.  However, because we have not
accounted for all the equivalences in Section~\ref{subsec:equiv} some pairs
in our list may be equivalent to each other.  In some sense such pairs
should not actually be considered distinct, so to count how many distinct
complex Golay pairs exist in length $n$ we would like
to find and remove pairs that are equivalent from the list.
Additionally, to verify the counts given by \cite{gibson2011quaternary}
it is necessary to produce a list that contains
\emph{all} complex Golay pairs.  We now describe an algorithm that
does both, i.e., it produces a list of all complex Golay pairs as 
well as a list of all inequivalent complex Golay pairs.

In detail, our algorithm performs the following steps:
\begin{enumerate}
\item Initialize $\Omega_\text{all}$ to be the empty set.  This variable will be a set that will be
populated with and eventually contain all complex Golay pairs of length~$n$.
\item Initialize $\Omega_{\text{inequiv}}$ to be the empty set.  This variable
will be a set that will be populated with and eventually contain all inequivalent
complex Golay pairs of length~$n$.
\item For each $(A,B)$ in the set of complex Golay pairs generated in stage 2:
\begin{enumerate}
\item If $(A,B)$ is already in $\Omega_{\text{all}}$ then skip this $(A,B)$
and proceed to the next pair.
\item Initialize $\Gamma$ to be the set containing $(A,B)$.
This variable will be a set that will be populated with and eventually
contain all complex Golay pairs equivalent to $(A,B)$.
\item For every $\gamma$ in $\Gamma$ add
$\operatorname{E1}(\gamma)$, $\dotsc$, $\operatorname{E5}(\gamma)$ to $\Gamma$.
Continue to do this until every pair in $\Gamma$ has been examined and no new
pairs are added to $\Gamma$.
\item Add $(A,B)$ to $\Omega_{\text{inequiv}}$ and add all pairs in $\Gamma$ to
$\Omega_{\text{all}}$.
\end{enumerate}
\end{enumerate}
Additionally, pseudocode for this algorithm is given in Algorithm~\ref{alg}.
\begin{algorithm}
\begin{algorithmic}[1]
\State $\Omega_\text{all}\coloneqq\emptyset$, $\Omega_\text{inequiv}\coloneqq\emptyset$
	\For{$\omega\in\{\,\text{list of generated complex Golay pairs}\,\}\setminus\Omega_\text{all}$}
		\State $\Gamma\coloneqq\{\omega\}$
		\For{$\gamma\in\Gamma$}
			\State $\Gamma\coloneqq\Gamma\cup\{\operatorname{E1}(\gamma),\dotsc,\operatorname{E5}(\gamma)\}$
		\EndFor
		\State $\Omega_\text{inequiv}\coloneqq\Omega_\text{inequiv}\cup\{\omega\}$
		\State $\Omega_\text{all}\coloneqq\Omega_\text{all}\cup\Gamma$
	\EndFor
\end{algorithmic}
\caption{Pseudocode for our postprocessing algorithm.
Note that the sets being iterated through will update during each iteration.}\label{alg}
\end{algorithm}

After running this algorithm listing the members of $\Omega_{\text{all}}$
gives a list of all complex Golay pairs of length~$n$ and
listing the members of $\Omega_{\text{inequiv}}$
gives a list of all inequivalent complex Golay pairs of length~$n$.
At this point we can also construct the complete list of sequences
that appear in any complex Golay pair of length~$n$.
To do this it suffices to add $A$ and $B$ to a new
set $\Omega_{\text{seqs}}$ for each $(A,B)\in\Omega_{\text{all}}$.

\subsection{Optimizations}\label{sec:optimization}

Although the method described will correctly enumerate all complex Golay
pairs of a given length $n$, for the benefit of potential implementors
we mention a few optimizations that we found helpful.

In the preprocessing step and stage~1 it is necessary to evaluate
a polynomial at points on the unit circle and determine its squared
absolute value.
The fastest way we found to do this used the discrete Fourier transform (DFT).
For example, let $A'$ be the sequence $\Aeven$, $\Aodd$, or $A$
under consideration but padded with trailing zeros
so that $A'$ is of length~$N$.
By definition of the discrete Fourier transform we have that
the $j$th entry of $\DFT(A')$
is exactly $A'(z_j)$ where $z_j\coloneqq\exp(2\pi ij/N)$.
Thus, we determine the values of $\abs{A'(z_j)}^2$ by taking
the squared absolute values of the entries of $\DFT(A')$.
If $\abs{A'(z)}^2>2n$ for some~$z$ then by Lemma~\ref{cor:filter}
or Corollary~\ref{cor:fiedler}
we can discard $A'$ from consideration.  To guard against potential
inaccuracies introduced by the algorithms used to compute the DFT we
actually ensure that $\abs{A'(z)}^2>2n+\epsilon$ for some tolerance~$\epsilon$
that is small but larger than the accuracy of the DFT
(e.g., $\epsilon=10^{-3}$).

In stage 1 we solve the quadratic Diophantine equation
$u^2 + v^2 + x^2 + y^2 = 2n$ in integers $u$, $v$, $x$, $y$.
In fact, we can also add the constraints
\[ u+v\equiv n \pmod{2} \quad\text{and}\quad x+y\equiv n\pmod{2} \]
(the second is unnecessary, being implied by the first) because of the following lemma.
\begin{lemma}
Suppose\/ $R$ and\/ $I$ are the $\resum$ and $\imsum$ of a sequence\/
$A\in\brace{\pm1,\pm i}^n$.  Then\/ $R+I\equiv n\pmod{2}$.
\end{lemma}
\begin{proof}
Let $\#_c$ denote the number of entries in $A$ with value $c$.  Then
\[ R + I = (\#_1 - \#_{-1}) + (\#_{i} - \#_{-i}) \equiv \#_1 + \#_{-1} + \#_{i} + \#_{-i} \pmod{2} \]
since $-1\equiv1\pmod{2}$.
The quantity on the right is $n$ since there are $n$ entries in $A$.
\end{proof}

In the final filtering step of stage~1 we ensure that Diophantine equations of the form
$R^2 + I^2 + x^2 + y^2 = 2n$
are solvable in integers $(x,y)$ where $R$ and $I$ are given.
This can be efficiently done by precomputing
a Boolean two dimensional array
$D$ such that $D_{\abs{R},\abs{I}}$ is true if and only if this equation 
has a solution, making the check for solvability a fast lookup.
At this point we also check if we can find a $z$ on the unit circle with $\abs{A(z)}^2>2n$.
Because the joining process is the bottleneck of our algorithm it is
important to make this filtering step as efficient as possible;
we did this by
computing the values of $\abs{A(z)}^2$ via the expression
$\abs{A_1(z)+A_2(z)}^2$ and precomputing the values of $A_1(z)$ and $A_2(z)$
for $A_1\in\Lodd$ and $A_2\in\Leven$ 
on points of the form $z=\exp(2\pi i j/32)$ with $j=0,\dotsc,31$.

We found that it was more efficient
to not check the condition for each $j$ in ascending order (i.e., for each~$z$
in ascending complex argument)
but to first perform the check on points $z$ with larger spacing between them.
In our implementation we checked the condition for $z$ of the form $\exp(2\pi ij/N)$
with $j$ odd and $N=8$, $16$, and $32$ in that order.
(This ignores checking the
condition when $z\in\{\pm1,\pm i\}$ but that is desirable since
$\abs{A(i^k)}^2=\resum(i^k\star A)^2+\imsum(i^k\star A)^2$
and the sums-of-squares condition is a strictly stronger filtering method.)

The downside of precomputing the values of $A_1(z)$ and $A_2(z)$
for $A_1\in\Lodd$ and $A_2\in\Leven$ is that when the lists $\Lodd$
and $\Leven$ are large this requires a significant amount of
memory.  In our implementation storing the values of $A_2(z)$
for all $A_2\in\Leven$ in the
lengths $27$ and $28$ each required about 8GB of RAM, the
searches in lengths $25$ and $26$ each required about 2GB of RAM,
and the searches in lengths $23$ and $24$ each required about 0.5GB of RAM.
However, the amount of precomputation could be increased or decreased
depending on the amount of memory available.

In stage~1 we make a pass through the lists $\Lodd$ and $\Leven$
for each valid possibility of the values $(u_0,u_1,u_2,u_3)$.
However, it is possible to reuse some work between passes.  For example,
it is only necessary to sort the lists $\Lodd$ and $\Leven$ once.
Even though the values of $V_2$ depend on the values $(u_0,u_1,u_2,u_3)$
the relative ordering of $\Leven$ is unaffected since all $V_2$ are
shifted by the same amount in each pass.  It is even possible to
do some parts of the passes simultaneously: for example, if the
first coordinate of $V_1$ is greater than the first coordinate of $V_2$
for some $u_0$ then $V_1$ is lexicographically greater than $V_2$
in all passes with that value of $u_0$ and so we can iterate to the
next $L_2$ in all such passes.

At the conclusion of stage~1 we added an additional round of filtering
to the sequences in $\LA$ that was not in our previous work.
Since this step was not
a bottleneck we used a more involved filtering process here;
in our implementation we ensured that $\abs{A(z)}^2\leq2n$
for the $2^{10}$ points of the form $z_j=\exp(2\pi ij/2^{10})$.  Additionally,
we keep track of the local maxima found by this process and use the
quadratic interpolation method described in Section~\ref{sec:preproc}
to find more accurate approximations to the local maxima of $\abs{A(z)}^2$.

In stage 2 one can also include properties that complex Golay sequences
must satisfy in the SAT instances.
As an example of this, we state the following proposition that
was published by \cite{bright2018enumeration} at ISSAC 2018.
\begin{proposition}\label{prop:prod}
Let\/ $(A,B)$ be a complex Golay pair.  Then
\[ a_k a_{n-k-1} b_k b_{n-k-1} = \pm 1 \qquad\text{for}\qquad\text{$k=0$, $\dotsc$, $n-1$}. \]
\end{proposition}

To prove this, we use the following simple lemma.
\begin{lemma}\label{lem:sumsofi}
Let\/ $c_k\in\Z_4$ for\/ $k=0$, $\dotsc$, $n-1$.  Then
\[ \sum_{k=0}^{n-1}i^{c_k}=0 \qquad\text{implies}\qquad \sum_{k=0}^{n-1} c_k\equiv0\pmod{2}. \]
\end{lemma}
\begin{proof}
Let $\#_c$ denote the number of $c_k$ with value $c$.
Note that the sum on the left implies that $\#_0=\#_2$ and $\#_1=\#_3$
because the $1$s must cancel with the $-1$s and the $i$s must cancel with the $-i$s.
Then
$\sum_{k=0}^{n-1}c_k=\#_1+2\#_2+3\#_3=4\#_1+2\#_2\equiv0\pmod{2}$.
\end{proof}
We now prove Proposition~\ref{prop:prod}.
\begin{proof}
Let $c_k$, $d_k\in\Z_4$ be such
that $a_k=i^{c_k}$ and $b_k=i^{d_k}$.  Using this notation
the multiplicative equation from Proposition~\ref{prop:prod} becomes
the additive congruence
\[ c_k + c_{n-k-1} + d_k + d_{n-k-1} \equiv 0 \pmod{2} . \label{eq:star}\tag{$2$} \]
Since $(A,B)$ is a complex Golay pair, the autocorrelation equations give us
\[ \sum_{k=0}^{n-s-1}\paren[\Big]{i^{c_k-c_{k+s}}+i^{d_k-d_{k+s}}} = 0 \]
for $s=1$, $\dotsc$, $n-1$.  Using Lemma~\ref{lem:sumsofi} and the fact that
$-1\equiv1\pmod{2}$ gives
\[ \sum_{k=0}^{n-s-1}\paren[\big]{c_k + c_{k+s} + d_k + d_{k+s}} \equiv 0 \pmod{2} \]
for $s=1$, $\dotsc$, $n-1$.
With $s=n-1$ (or $s=1$) one immediately derives~\eqref{eq:star} for $k=0$.
Adding together these congruences for $s=n-1$ and $s=n-2$ 
derives~\eqref{eq:star} for $k=1$.
The congruences for $s=n-2$ and $s=n-3$ give~\eqref{eq:star} for $k=2$ and
proceeding in this manner one derives~\eqref{eq:star} for all $k$.
\end{proof}

In short, Proposition~\ref{prop:prod} tells us that an even number of
$a_k$, $a_{n-k-1}$, $b_k$, and $b_{n-k-1}$ are real for each $k=0$, $\dotsc$, $n-1$.
For example, if exactly one of $a_k$ and $a_{n-k-1}$ is real then exactly
one of $b_k$ and $b_{n-k-1}$ must also be real.
In this case, using our encoding from
Section~\ref{sec:stage2} we add the two binary clauses
\[ v_{2k}\lor v_{2(n-k-1)} \qquad\text{and}\qquad \lnot v_{2k}\lor \lnot v_{2(n-k-1)} \]
to our SAT instance.
These clauses say that exactly one of $v_{2k}$ and $v_{2(n-k-1)}$ is true.
Conversely, if an even number of $a_k$ and $a_{n-k-1}$ are real then
an even number of $b_k$ and $b_{n-k-1}$ must also be real.  In this case we
add the two binary clauses
\[ v_{2k}\lor \lnot v_{2(n-k-1)} \qquad\text{and}\qquad \lnot v_{2k}\lor v_{2(n-k-1)} \]
to our SAT instance.

\section{Results}\label{sec:results}

In order to provide a verification of the counts given by \cite{fiedler2013small},
\cite{gibson2011quaternary}, and \cite*{CHK:DM:2002}
we implemented the enumeration method described in Section~\ref{sec:method}.
The preprocessing step was performed by a C program
and used the mathematical library FFTW by \cite{frigo2005design} for computing
the values of $f_0$, $\dotsc$, $f_{127}$ as described in Section~\ref{sec:optimization}. 
Stage~1 was performed by a C++ program,
used FFTW for computing the values of $A_1(z)$ and $A_2(z)$,
and a \textsc{Maple} script by \cite{nsoks}
for determining the solvability of the Diophantine equations 
given in Section~\ref{sec:stage2}.  Stage~2 was performed
by the programmatic SAT solver \textsc{MapleSAT} by \cite{liang2017empirical}.
The postprocessing step was performed by a Python script.

\begin{table}\begin{center}\begin{tabular}{cccc}
& \multicolumn{3}{c}{Total CPU Time in hours} \\
$n$ & Preproc. & Stage~1 & Stage~2 \\ \hline
17         & 0.00       & 0.00       & 0.00       \\ 
18         & 0.00       & 0.01       & 0.01       \\ 
19         & 0.00       & 0.01       & 0.01       \\ 
20         & 0.00       & 0.08       & 0.04       \\ 
21         & 0.00       & 0.71       & 0.13       \\ 
22         & 0.00       & 0.88       & 0.15       \\ 
23         & 0.01       & 4.67       & 0.19       \\ 
24         & 0.01       & 7.09       & 1.24       \\ 
25         & 0.03       & 97.86      & 1.37       \\ 
26         & 0.06       & 234.43     & 4.72       \\ 
27         & 0.12       & 3255.56    & 49.56      \\
28         & 0.22       & 2543.34    & 25.73      
\end{tabular}\end{center}
\caption{The time used to run the various stages of our algorithm in lengths $17\leq n\leq28$.}
\label{tbl:timings}\end{table}

We ran our implementation on a cluster of machines running CentOS 7 and
using Intel Xeon E5-2683V4 processors running at $2.1$~GHz and
using at most 8.5GB of RAM.
The lengths up to $22$ were run on a single core,
the lengths $23$ and $24$ were run on 10 cores,
the lengths $25$ and $26$ were run on 100 cores, and the lengths
$27$ and $28$ were run on 1000 cores.
The work was parallelized across~$N$ cores by partitioning
the list $\Lodd$ into~$N$ sublists of approximately equal size.
Everything in the stages proceeded exactly as before
except that in stage~1 the list $\Lodd$
was about~$N$ times shorter than it would have been otherwise.
In practice, this allowed us to complete the search nearly~$N$ times faster;
for example, our search in length $28$ required about 3.5 months of CPU time
but completed in about $3$ hours when distributed across 1000 cores.
The timings for the preprocessing step and the two stages of our algorithm are
given in Table~\ref{tbl:timings}; the timings for the postprocessing step
were negligible.  The times are given as the total amount of CPU time used
across all cores.
Our code is available online as a part of the \textsc{MathCheck} project
(see \href{https://uwaterloo.ca/mathcheck}{\nolinkurl{uwaterloo.ca/mathcheck}})
and the resulting enumeration of complex Golay pairs 
is available from \url{https://doi.org/10.5281/zenodo.1246337}.

The sizes of the lists $\Leven$ and $\Lodd$ computed in the preprocessing step
and the size of the list $\LA$ computed in stage~1
are given in Table~\ref{tbl:listsizes} for all lengths $n\leq28$.
Without applying any filtering $\LA$ would have size $4^n$ so
Table~\ref{tbl:listsizes} demonstrates the power of the criteria we used to
perform filtering; for example in length $28$ less than $10^{-8}\%$
of the possible sequences $A$ were added to $\LA$.
A result of \cite{gersho1979coefficient} implies that for~$z$ on the unit circle
the maximum value of $\abs{A(z)}^2$
is $n\log n + O(n\log\log n)$
for almost all $\brace{\pm1,\pm i}$-polynomials~$A$ of degree~$n$.
In particular, the number of polynomials that do not satisfy
this is at most $4^n/(\log n)^4$, implying that
the size of $\LA$ will be $o(4^n)$
though in practice we even have $\abs{\LA}\leq2^n$.

The generated SAT instances had $2n$ variables
(encoding the entries $b_0$, $\dotsc$, $b_{n-1}$),
two unit clauses (encoding $b_0=1$),
$2\floor{n/2}$ binary clauses (encoding Proposition~\ref{prop:prod}), and
$n-1$ programmatic constraints (encoding Definition~\ref{def:cgp}).
Once the values of
$b_k$ are known for $k=0$, $\dotsc$, $\floor{n/2}$,
the programmatic constraints uniquely determine the remaining values of $b_k$.
Thus $4^{n/2}$ is a crude upper bound on the number
of possible assignments to the values of $b_0$, $\dotsc$, $b_{n-1}$.

We now give an upper bound on the complexity of the parts of our algorithm that do
not run in polynomial time.  The preprocessing and postprocessing costs
are dominated by the costs of stage~1 and~2. The cost of stage~1
is at most $4^n$, an upper bound on the number of generated SAT instances.
The cost of stage~2 is difficult to specify rigorously since it may
depend on the specific SAT solver that is used.  However, there are
at most~$4^{n/2}$ possible assignments for each instance, so using a solver
that checks each assignment once gives an overall cost of $4^{3n/2}$.
The running time is better in practice, though seemingly still exponential.

Finally, we provide counts of the total number of complex Golay pairs of
length $n\leq28$ in Table~\ref{tbl:counts}.
The sizes of $\Omega_{\text{seqs}}$ match those given
by \cite{fiedler2013small} in all cases,
the sizes of $\Omega_{\text{all}}$
match those given by \cite{gibson2011quaternary} for all $n\leq26$
(the largest length they included)
and the sizes of $\Omega_{\text{inequiv}}$
match those given by \cite{CHK:DM:2002}
for all $n\leq19$ (the largest length they exhaustively solved).

\begin{table}\begin{center}\begin{tabular}{cccc}
$n$ & $\abs{\Leven}$ & $\abs{\Lodd}$ & $\abs{\LA}$ \\ \hline
1          & 1          & $-$        & 1          \\ 
2          & 3          & 1          & 3          \\ 
3          & 3          & 1          & 1          \\ 
4          & 3          & 4          & 3          \\ 
5          & 12         & 4          & 5          \\ 
6          & 12         & 16         & 14         \\ 
7          & 39         & 16         & 12         \\ 
8          & 48         & 64         & 36         \\ 
9          & 153        & 64         & 44         \\ 
10         & 153        & 204        & 118        \\ 
11         & 561        & 252        & 99         \\ 
12         & 645        & 860        & 445        \\ 
13         & 2121       & 884        & 279        \\ 
14         & 2463       & 3284       & 294        \\ 
15         & 8340       & 3572       & 1650       \\ 
16         & 9087       & 12116      & 829        \\ 
17         & 31275      & 12824      & 3233       \\ 
18         & 34560      & 46080      & 11159      \\ 
19         & 117597     & 50944      & 10918      \\ 
20         & 130215     & 173620     & 26876      \\ 
21         & 446052     & 194004     & 81941      \\ 
22         & 500478     & 667304     & 90163      \\ 
23         & 1694871    & 732232     & 118747     \\ 
24         & 1886562    & 2515416    & 200138     \\ 
25         & 6447250    & 2727452    & 709584     \\ 
26         & 7183879    & 9578506    & 737891     \\ 
27         & 24426370   & 10591928   & 7618474    \\ 
28         & 27265578   & 36354113   & 3687209    
\end{tabular}\end{center}
\caption{The number of sequences $\Aeven$, $\Aodd$,
and $A$ that passed the filtering conditions of our algorithm
in lengths up to~$28$.}
\label{tbl:listsizes}\end{table}

\begin{table}\begin{center}\begin{tabular}{cccc}
$n$ & $\abs{\Omega_{\text{seqs}}}$ & $\abs{\Omega_{\text{all}}}$ & $\abs{\Omega_{\text{inequiv}}}$ \\ \hline
1          & 4          & 16         & 1          \\ 
2          & 16         & 64         & 1          \\ 
3          & 16         & 128        & 1          \\ 
4          & 64         & 512        & 2          \\ 
5          & 64         & 512        & 1          \\ 
6          & 256        & 2048       & 3          \\ 
7          & 0          & 0          & 0          \\ 
8          & 768        & 6656       & 17         \\ 
9          & 0          & 0          & 0          \\ 
10         & 1536       & 12288      & 20         \\ 
11         & 64         & 512        & 1          \\ 
12         & 4608       & 36864      & 52         \\ 
13         & 64         & 512        & 1          \\ 
14         & 0          & 0          & 0          \\ 
15         & 0          & 0          & 0          \\ 
16         & 13312      & 106496     & 204        \\ 
17         & 0          & 0          & 0          \\ 
18         & 3072       & 24576      & 24         \\ 
19         & 0          & 0          & 0          \\ 
20         & 26880      & 215040     & 340        \\ 
21         & 0          & 0          & 0          \\ 
22         & 1024       & 8192       & 12         \\ 
23         & 0          & 0          & 0          \\ 
24         & 98304      & 786432     & 1056       \\ 
25         & 0          & 0          & 0          \\ 
26         & 1280       & 10240      & 16         \\ 
27         & 0          & 0          & 0          \\ 
28         & 0          & 0          & 0          
\end{tabular}\end{center}
\caption{The number complex Golay pairs in lengths up to~$28$.
The table counts the number of individual sequences, the
number of pairs, and the number of pairs up to equivalence.}
\label{tbl:counts}\end{table}

\begin{table}\begin{center}\begin{tabular}{cccc}
& \multicolumn{3}{c}{Total CPU Time in hours} \\
$n$ & ISSAC'18 & This paper \\ \hline
17 & 0.07 & 0.00 \\
18 & 0.27 & 0.02 \\ 
19 & 0.26 & 0.02 \\
20 & 0.80 & 0.12 \\
21 & 3.86 & 0.84 \\
22 & 10.77 & 1.03 \\
23 & 45.18 & 4.87 \\
24 & 86.97 & 8.34 \\
25 & 702.39 & 99.26 \\
26 & $-$ & 239.21 \\
27 & $-$ & 3305.24 \\
28 & $-$ & 2569.29 
\end{tabular}\end{center}
\caption{A comparison between the method presented at ISSAC 2018~\citep{bright2018enumeration}
and the method used in this paper.  The times measure the total amount
of computation time that was used to run a complete search in the lengths $17\leq n\leq28$
when run on the same hardware.}
\label{tbl:oldtimings}\end{table}

Because \cite{fiedler2013small}, \cite{gibson2011quaternary}, and \cite{CHK:DM:2002} do not provide
implementations or timings for the enumerations they completed it is not possible
for us to compare the efficiency of our algorithm to previous algorithms.
However, in Table~\ref{tbl:oldtimings} we compare our implementation's timings
to the timings we previously presented~\citep{bright2018enumeration}.
Table~\ref{tbl:oldtimings} shows that the improved version of our algorithm performs about
an order of magnitude faster in general.

\section{Future Work}

Besides increasing the length to which complex Golay pairs have been enumerated
there are a number of avenues for improvements that could be made in
future work.  As one example, we remark that
we have not exploited the algebraic structure
of complex Golay pairs revealed by \cite*{CHK:DM:2002}.
In particular, their work contains a theorem implying that if $p\equiv3\pmod{4}$
is a prime dividing~$n$ and~$A$ is a member of
a complex Golay pair of length $n$ then $A(z)$ is not irreducible
over $\FF_p(i)$.  Ensuring that this property holds
could be added to the filtering conditions used in stage~1.
In fact, the authors relate the factorization of $A(z)$ over $\FF_p(i)$
to the factorization of $B(z)$ over $\FF_p(i)$ for any complex Golay pair
$(A,B)$.  This factorization could potentially be used to perform stage~2
more efficiently, possibly supplementing or replacing the SAT solver entirely,
though it is unclear if such a method would perform better than our method in practice.
In any case, it would not be possible to apply their theorem in all lengths
since the length $n$ might only be divisible by primes $p=2$ or $p\equiv1\pmod{4}$.

A second possible improvement could be to use the symbolic
form of $f(\theta)$ defined in Section~\ref{sec:preproc}
to help find the maximum of $f(\theta)$.  For example,
a consequence of Lemma~\ref{lem:normeval} and Euler's identity is that
\[ f(\theta) = N_{A'}(0) + 2\sum_{s=1}^{n-1}\Re(N_{A'}(s))\cos(s\theta) + 2\sum_{s=1}^{n-1}\Im(N_{A'}(s))\sin(s\theta) \]
showing that $f$ has a simple form as a sum of sines and cosines.
Finding the real roots of the derivative of $f$
would reveal the locations of the local maxima and minima of $f$.
However, again it is unclear if this method would perform
better than our approximation method in practice.
Note that our method is not guaranteed to find
the global maximum of $f$ since we only apply the quadratic
interpolation optimization method to ranges that we
know contain a local maximum by examining the
initial sampling of $f$.  With a more careful branch-and-bound
algorithm it would be possible to guarantee that the global maximum
was found but the approximations that we found were effective enough
that we did not pursue this.

Another possible improvement could be obtained by deriving further properties
like Proposition~\ref{prop:prod} that complex Golay pairs must satisfy.
For example, consider the following property
that could be viewed as a strengthening of Proposition~\ref{prop:prod}:
\[ a_k\overline{a_{n-k-1}} = (-1)^{n+1}b_k\overline{b_{n-k-1}}
\qquad\text{for}\qquad\text{$k=1$, $\dotsc$, $n-2$}. \]
An examination of all complex Golay pairs up to length~$28$ reveals that
they all satisfy this property except for a \emph{single} complex Golay pair
up to equivalence. 
The only pairs that don't satisfy this property are equivalent to 
\[ \paren[\big]{[1, 1, 1, -1, 1, 1, -1, 1], [1, i, i, -1, 1, -i, -i, -1]} \]
and were already singled out by \cite*{fiedler2008multi} for being
the only known examples of what they call
``cross-over'' Golay sequence pairs.
Since a counterexample exists to this property there is no hope of proving
it in general, but perhaps a suitable generalization could be proven.

Lastly, the running time analysis of our algorithm could be improved.
We have given a crude upper bound of $4^{3n/2}=8^n$ for the complexity
of the parts of our algorithm that do not run in polynomial time.  However,
this running time would be improved by tighter bounds on either
the size of $\LA$ or the number of possible assignments that the SAT solver will try.

\section*{Acknowledgements}

This work was made possible by the facilities of the Shared Hierarchical 
Academic Research Computing Network (SHARCNET) and Compute/Calcul Canada.
The authors would also like to thank the anonymous reviewers
for their comments on this paper and on our ISSAC 2018 submission.
Their careful reviews improved our work's clarity and presentation.

\bibliography{jsc-cgolay}

\end{document}